\documentclass{article}
\usepackage[T1]{fontenc}
\usepackage[utf8]{inputenc}

\usepackage{amsmath}
\usepackage{amssymb}
\usepackage{hyperref}
\usepackage{amsthm}
\usepackage{amscd}
\usepackage{graphics}
\usepackage{soul}
\usepackage{thm-restate}
\usepackage[dvipsnames]{xcolor}
\usepackage{fouriernc}
\usepackage{yfonts}

\newcommand{\norm}[1]{\left|\left|#1\right|\right|}

\newcommand{\dom}{\operatorname{dom}}

\newcommand{\R}{\mathbb{R}}
\newcommand{\C}{\mathbb{C}}

\newcommand{\N}{\mathbb{N}}

\newcommand{\NN}{\N^{\N}}
\newcommand{\rank}{\operatorname{rank}}

\renewcommand{\O}{\mathcal{O}}
\newcommand{\A}{\mathcal{A}}
\newcommand{\V}{\mathcal{V}}

\newcommand{\K}{\mathcal{K}}

\newlength\arrowheight

\newcommand{\Set}[2]{\left\{#1 \mid #2\right\}}

\newcommand{\vlt}{\vartriangleleft}
\newcommand{\vleq}{\trianglelefteq}
\newcommand{\vgt}{\vartriangleright}
\newcommand{\vgeq}{\trianglerighteq}

\theoremstyle{definition}
\newtheorem{theorem}{Theorem}

\newtheorem{corollary}[theorem]{Corollary}
\newtheorem{proposition}[theorem]{Proposition}

\newtheorem{lemma}[theorem]{Lemma}

\begin{document}

\title{Termination of Real Linear Loops}

\author{
Eike Neumann
\qquad
Margret Tembo\\\\
Swansea University, UK\\
\texttt{\{e.f.neumann, m.k.tembo\}@swansea.ac.uk}
}
\date{}

\maketitle

\begin{abstract}
We study the problem of deciding universal termination of linear and affine loops over the reals in the bit-model of real computation.
We show that both problems are as close to decidable as one can expect them to be:
there exist sound partial algorithms that halt on all problem instances whose answer is robust under all sufficiently small perturbations.
We further show that in each case the set of non-robust problem instances has Lebesgue measure zero. 
\end{abstract}

\section{Introduction}

Discrete-time linear systems arise naturally in a wide range of application areas across systems and control theory and verification.
A fundamental liveness verification task is to decide whether every system state in a given set eventually escapes the set under the dynamics.
Here we consider the following instantiation of this problem: does every point in a given polyhedron $P$ escape $P$ under the iteration of a given matrix $A$?
Tiwari \cite{Tiwari} has shown that this problem is decidable when $A$ and $P$ are specified by exact rational data.

In applications, the system under consideration may not always be specifiable by exact rational data -- for example, it may not be possible to assume that the transition probabilities of a Markov chain be exactly known.
This leads to the question whether the problem remains decidable when the matrix or polyhedron are known only approximately, but with guaranteed error bounds.
Following \cite{DPLRSR21}, we model this situation by studying the problem for matrices and polyhedra specified by real data in the sense of computable analysis \cite{WeihrauchBook,CollinsMonograph,BrattkaHertling21}.  This line of work also connects to recent developments in exact real computation and program verification, such as the Hoare-logic framework of Park et al.~\cite{ParkEtAl2024}, which provides a natural setting for
verification problems over real data.
Further, Tiwari's method in general requires exact computation with algebraic numbers, in particular symbolic computation of the Jordan normal form, which is computationally very expensive (cf.~the discussion in \cite[p.80]{Tiwari}). We will solve the real-number version of the problem using validated numerical methods, which are potentially much more efficient.

Let us define the decision problem we consider more precisely.
For real vectors $v, w \in \R^m$ of the same dimension, we write $v \vgt w$ if $v_j > w_j$ for all $j \in \{1,\dots,m\}$.
The relations $v \vgeq w$, $v \vlt w$, and $v \vleq$ are defined analogously.
For a matrix $B \in \R^{m \times n}$, we let $P(B) = \Set{x \in \R^n}{B x \vgt 0}$ denote the open polyhedron defined by $B$.
We let $\overline{P}(B) = \Set{x \in \R^n}{B x \vgeq 0}$ denote the closed polyhedron defined by $B$.
Observe that, since we do not impose any constraints on $B$, the set $\overline{P}(B)$ is not necessarily equal to the closure of 
$P(B)$.
The Linear Universal Escape Problem is the following decision problem:
given matrices $A \in \R^{n \times n}$ and $B \in \R^{m \times n}$, decide whether for all $x \in \R^n$ there exists $k \in \N$ with $A^k x \notin P(B)$.

Since computability implies continuity, no non-trivial subset of Euclidean space has a computable characteristic function.
Thus, we cannot hope to decide any interesting problem over real data in the na\"ive sense of computing its characteristic function.
However, we can ask for the next best thing:
A subset $S$ of a represented space $X$ is called \emph{maximally partially decidable} \cite{DPLRSR21} if there exists an algorithm which takes as input an element $x \in X$, and either runs forever or halts in finite time. 
If it halts, it must report correctly whether $x$ belongs to $S$. 
Moreover, it is required to halt on all inputs $x \notin \partial S$.
Observe that if $X$ is discrete in the sense that equality for elements of $X$ is semidecidable, then a subset of $X$ is decidable in the usual sense if and only if it is maximally partially decidable. 
Thus, maximal partial decidability generalises decidability over discrete spaces -- generally in a more meaningful way than the na\"ive definition.
See \cite{DPLRSR21} for further discussion and motivation of this concept.
When we think of $S \subseteq X$ as a decision problem, we also refer to points $x \in X$ as \emph{instances of $S$}.
A point $x \in \partial S$ is called a \emph{boundary instance}.
A point $x \in X \setminus \partial S$ is called a \emph{robust instance}.

Our first main result is the following:

\begin{theorem}\label{Theorem: main theorem}
    The Linear Universal Escape Problem is maximally partially decidable.
\end{theorem}

Tiwari's decision method \cite{Tiwari} for the universal escape problem for rational inputs relies on the computation of the Jordan normal form, which is well-known to be uncomputable for real matrices for continuity reasons (cf.~\cite{BrattkaZieglerLA}).
Hence, we pursue an entirely different approach:
We characterise the robust problem instances via first-order formulas, where all universal quantification ranges over compact sets, and all existential quantification ranges over overt sets. This immediately yields a maximal partial algorithm.

To state our characterisation results, let us introduce some notation.
For a matrix $A \in \R^{n \times n}$, we let $\sigma(A) \subseteq \C$ denote its spectrum.
For a real number $r$, we let $\sigma_{\geq r}(A) = \sigma(A) \cap [r, +\infty)$
and $\sigma_{> r}(A) = \sigma(A) \cap (r, +\infty)$.
We let $\sigma^{\text{odd}}_{> r}(A)$ denote the eigenvalues in $\sigma_{> r}(A)$ with odd algebraic multiplicity. We let $A_j$ denote the $j^{\text{th}}$ row of $A$. Let $S^{n-1} \subseteq \R^n$ denote the (Euclidean) unit sphere.
We call an instance \((A,B)\) \emph{trapped} if some \(x\in P(B)\) satisfies
\(A^k x\in P(B)\) for all \(k\geq 0\) and \emph{escaping} otherwise. Robust escaping and robust trapped
instances are those lying in the interiors of the escaping and trapped sets,
respectively.

\begin{theorem}\label{Theorem: characterisation of robust linear instances}
    Let $\left(A, B\right) \in \R^{n \times n} \times \R^{m \times n}$ be an instance of the Linear Universal Escape Problem.
    Then:
    \begin{enumerate}
        \item \label{Item: robust linear escaping} $(A,B)$ is a robust escaping instance if and only if it satisfies the formula
            \begin{equation}\label{eq: robust escaping condition}
            \forall\lambda\in\sigma_{\geq 0}(A).\forall v\in S^{n-1}.
            \left(
            A v=\lambda v\;\Longrightarrow\;
            \exists j\in\{1,\dots,m\}. B_j v<0
            \right).
            \end{equation}
        \item \label{Item: robust linear trapped} $(A,B)$ is a robust trapped instance if and only if it satisfies the formula
            \begin{equation}\label{eq: robust trappedness condition}
                \exists \lambda\in\sigma^{\text{odd}}_{>0}({A}). \forall v\in S^{n-1}.
                \left(
                A v=\lambda v\;\Longrightarrow\;
                \left(B v\vgt0 \lor B v\vlt 0\right)
                \right).
            \end{equation}
    \end{enumerate}
\end{theorem}

So far, we have focused on (homogeneous) linear systems.
It is natural to consider the same problem for affine systems.
The Affine Universal Escape Problem is the following decision problem:
Given $A \in \R^{n \times n}$, $b \in \R^n$, $B \in \R^{m \times n}$, $\eta \in \R^m$,
decide whether every point escapes the affine polyhedron $P(B,\eta) = \Set{x \in \R^n}{Bx \vgt \eta}$ under iteration of the map 
$f(x) = Ax + b$ in the sense that for all $x \in \R^n$ there exists $k \in \N$ with $f^k(x) \notin P(B,\eta)$.
Over discrete data, the Linear Universal Escape Problem is a special case of the Affine Universal Escape Problem.
Further, the Affine Universal Escape Problem reduces straightforwardly to its homogeneous linear counterpart via homogenisation.
Explicitly, given an instance $(A, B, b, \eta)$ of the Affine Universal Escape Problem, we can effectively compute the homogenised instance $\left(\widehat{A}, \widehat{B}\right) \in \R^{(n+1)\times(n+1)} \times \R^{(m+1) \times (n+1)}$, given as
\[
    \widehat{A} = 
    \begin{pmatrix}
        A   & b \\
        0   & 1     
    \end{pmatrix},
    \;
    \widehat{B} = 
    \begin{pmatrix}
        B          & -\eta \\
        0          & 1
    \end{pmatrix}.
\]
It is easy to see that $(A, B, b, \eta)$ escapes if and only if $\left(\widehat{A}, \widehat{B}\right)$ escapes.
However, in order for this mapping to be a reduction in the sense of maximal partial decidability, it needs to map robust instances of the Affine Universal Escape Problem to robust instances of the Linear Universal Escape Problem.
This turns out not to be the case, even for one-dimensional instances\footnote{Take for example $A = 1/2$, $b = -1$, $B = 1$, and $\eta = 0$.}.
Similarly, it is easy to see that there exist robust instances $(A,B)$ of the Linear Universal Escape Problem such that the instance $(A, 0, B, 0)$ is not a robust instance of the Affine Universal Escape Problem\footnote{Take for example $A = 1/2$ and $B = 1$.}.
Thus, while the two problems are closely related, there is no obvious reduction
in either direction when it comes to maximal partial decidability.
Nonetheless, we obtain:

\begin{theorem}\label{Theorem: main theorem 2}
    The Affine Universal Escape Problem is maximally partially decidable.
\end{theorem}

Theorem \ref{Theorem: main theorem 2} is based on a characterisation of the robust instances, similar to Theorem \ref{Theorem: characterisation of robust linear instances}:

\begin{theorem}\label{Theorem: characterisation of robust affine instances}
    Let $(A, B, b, \eta) \in \R^{n \times n} \times \R^{m \times n} \times \R^n \times \R^m$ be an instance of the Affine Universal Escape Problem.
    Then:
    \begin{enumerate}
        \item \label{Item: robust affine escaping} $(A, B, b, \eta)$ is a robust escaping instance if and only if it satisfies the formula
        \begin{equation}\label{eq: robust affine escape}
            \forall \lambda \in \sigma_{\geq 1}\left(\widehat{A}\right).
            \forall \widehat{v} \in S^{n}. 
            \left(
                \widehat{A}\widehat{v} = \lambda \widehat{v}
                \;\Longrightarrow\;
                \exists j\in\{1,\dots,m+1\}. \widehat{B}_j \widehat{v} < 0
            \right).
        \end{equation}
        \item \label{Item: robust affine trapped} $(A, B, b, \eta)$ is a robust trapped instance if and only if it satisfies the formula
        \begin{align}\label{eq: robust affine trapped}
            \begin{split}
            &\left( 1 \notin \sigma(A) \land B \cdot (A - I)^{-1} b \vlt -\eta \right)\\
            &\lor
            \exists \lambda\in\sigma^{\text{odd}}_{> 1}(A).
            \forall v\in S^{n-1}. 
            \left(
            A v=\lambda v\;\Longrightarrow\;
            \left(B v\vgt0 \lor B v\vlt0\right)
            \right).
            \end{split}
        \end{align}
    \end{enumerate}
\end{theorem}
Observe that \eqref{eq: robust affine trapped} contains a slight abuse of notation: the matrix $(A - I)^{-1}$ is well defined if and only if $1 \notin \sigma(A)$. Thus, the second conjunct in the first disjunctive clause cannot be evaluated independently of the first conjunct.

Since our decision methods cannot halt on every input, it is natural to investigate the size of their halting sets.
We observe that our algorithms halt on almost every input:

\begin{theorem}\label{Theorem: measure of boundary instances}\hfill
    \begin{enumerate}
        \item The set of boundary instances of the Linear Universal Escape Problem has Lebesgue measure zero.
        \item The set of boundary instances of the Affine Universal Escape Problem has Lebesgue measure zero.
    \end{enumerate}
\end{theorem}

The rest of the paper is structured as follows:
In Section \ref{Section: Preliminaries} we summarise required results and definitions from computable analysis and linear systems.
In Section \ref{Section: Computability} we show that the formulas in Theorems \ref{Theorem: characterisation of robust linear instances} and \ref{Theorem: 
characterisation of robust affine instances} are semidecidable. This proves Theorems \ref{Theorem: main theorem} and \ref{Theorem: main theorem 2} relative to Theorems \ref{Theorem: characterisation of robust linear instances} and \ref{Theorem: characterisation of robust affine instances}.
The rest of the paper is dedicated to the proofs of Theorems \ref{Theorem: characterisation of robust linear instances}, \ref{Theorem: characterisation of robust affine instances}, and \ref{Theorem: measure of boundary instances}.
In Section \ref{Section: perturbation lemmas} we prove a few technical lemmas that are required for the proofs of Theorems \ref{Theorem: characterisation of robust linear instances} and \ref{Theorem: characterisation of robust affine instances}.
The first part of Theorem \ref{Theorem: characterisation of robust linear instances} is proved in Section \ref{Section: Robust Escape Linear}.
The second part is proved in Section \ref{Section: Robust Trapped Linear}.
Theorem \ref{Theorem: characterisation of robust affine instances} is proved in Sections \ref{Section: Robust Escape Affine} and \ref{Section: Robust Trapped Affine}.
Theorem \ref{Theorem: measure of boundary instances} is proved in Section \ref{Section: Measure Boundary}.
\section{Preliminaries}
\label{Section: Preliminaries}

\subsection{Represented Spaces}

We assume some familiarity with the basic results and definitions of computable analysis and the theory of represented spaces.
For details we refer to \cite{PaulyRepresented} or \cite{CollinsMonograph}.
A \emph{represented space} is a set $X$ together with a partial surjective map $\delta_X \colon \subseteq \NN \to X$.
Examples of represented spaces include computable metric spaces with the Cauchy representation \cite{BrattkaPresser,Ko,WeihrauchBook}.
A map $f \colon X \to Y$ is called \emph{computable} if there exists a computable function 
$F \colon \dom \delta_X \to \dom \delta_Y$ satisfying 
$\delta_Y \circ F = f \circ \delta_X$.
It is called \emph{continuous} if there exists a continuous such function.
Clearly, computability implies continuity.
Continuity in this sense is closely linked to topological continuity with respect to the final topology induced by the representation.
See \cite{PaulyRepresented,SchroederAdmissibility,SchroederPhD} for details.

Represented spaces enjoy excellent closure properties. In particular, they admit countable products and exponentials.
This is the basis for important hyperspace constructions.
Let $\Sigma = \{\top, \bot\}$ denote Sierpinski space, with the representation $\delta_{\Sigma} \colon \NN \to \Sigma$
where $\delta_{\Sigma}(0^\omega) = \bot$ and $\delta_{\Sigma}(x) = \top$ for all $x \in \NN \setminus \{0^\omega\}$.
Then, for a represented space $X$, the space $\O(X)$ of open subsets of $X$ is given by identification of an open set with its characteristic function as an element of the space $\Sigma^X$ of all continuous functions from $X$ to $\Sigma$. The space $\A(X)$ of closed subsets of $X$ is given by identifying a closed set $A$ with its complement as an element of $\O(X)$.

A subset $K \subseteq X$ is \emph{compact} if and only if the set $\square K = \Set{U \in \O(X)}{K \subseteq U}$ is an open subset of $\O(X)$.
The set $K$ is called \emph{computably compact} if $\square K$ is a computable open subset of $\O(X)$.
Dually a subset $A \subseteq X$ is \emph{overt} if and only if the set $\lozenge A = \Set{U \in \O(X)}{A \cap U \neq \emptyset}$ is an open subset of $\O(X)$.
The set $A$ is called \emph{computably overt} if $\lozenge A$ is a computable open subset of $\O(X)$.
It is easy to see that \emph{every} set is overt, so that only computable overtness yields a (classically) meaningful notion.

For a represented space $X$ the \emph{space $\K(X)$ of compact subsets of $X$} is the space whose underlying set is the set of saturated compact subsets of $X$ and whose representation is obtained by identifying a saturated compact set $K$ with the set $\square K \in \O(\O(X))$.
The restriction to \emph{saturated} compact sets is necessary, since the set $\square K$ determines $K$ only up to saturation.
Dually,  the \emph{space $\V(X)$ of overt subsets of $X$} is the space whose underlying set is the set of closed subsets of $X$ and whose representation is obtained by identifying a closed set $A$ with the set $\lozenge A \in \O(\O(X))$.
The restriction to \emph{closed} sets is necessary, since the set $\lozenge A$ determines $A$ only up to closure.
See \cite{PaulyRepresented,CollinsMonograph} for effective closure properties of compact and overt sets.

The following important proposition follows immediately from the definition:

\begin{proposition}\label{Proposition: overtness and compactness}
    Let $X$ be a represented space.
    Given compact sets $K_1,\dots,K_n \in \K(X)$,
    overt sets
    $A_1,\dots,A_n \in \V(X)$,
    and an open set $U \in \O(X^{2n})$,
    we can semidecide if the sentence
    \[
        \forall x_1 \in K_1. \exists y_1 \in V_1. \dots \forall x_n \in K_n. \exists y_n \in V_n.
                (x_1,y_1,\dots,x_n,y_n) \in U
    \]
    holds true.
\end{proposition}

\subsection{Linear Dynamical Systems}

We require two results on the long-term behaviour of linear dynamical systems.
We will decompose the systems under consideration into systems with only positive real eigenvalues and systems without any positive real eigenvalues.
For our purpose, the following theorem characterises the behaviour of the latter:

\begin{theorem}[Bell and Gerhold, \cite{BellGerhold07}]\label{Theorem: LRS without positive real eigenvalues}
    Let $f \colon \N \to \R$ be a linear recurrence sequence without positive real eigenvalues.
    Then $f$ is either eventually zero or $f$ assumes infinitely many (strictly) negative values.
\end{theorem}

\begin{corollary}\label{Corollary: matrix iterations without positive real eigenvalues}
    Let $E$ be a finite-dimensional real vector space.
    Let $\mathcal{L} \colon E \to E$ be a linear map without positive real eigenvalues.
    Let $\varphi \colon E \to \R$ be a linear functional.
    Then for all $x \in E$ there exist infinitely many $k \geq 0$ with 
    $\varphi \circ \mathcal{L}^k x \leq 0$.
\end{corollary}
\begin{proof}
    It follows from the Cayley-Hamilton theorem that the sequence of real numbers $\varphi \circ \mathcal{L}^k x$ satisfies a linear recurrence whose characteristic polynomial is that of the matrix $\mathcal{L}$.
    The result now immediately follows from Theorem \ref{Theorem: LRS without positive real eigenvalues}.
\end{proof}

The long-term behaviour of linear dynamical systems with only positive real eigenvalues is captured by the next lemma:

\begin{lemma}\label{Lemma: linear map with positive real eigenvalues}
    Let $E$ be a finite-dimensional real vector space.
    Let $\mathcal{L} \colon E \to E$ be a linear map with $\sigma(\mathcal{L}) \subseteq (0, +\infty)$.
    Let $\rho_1 < \dots < \rho_N$ denote the eigenvalues of $\mathcal{L}$.
    Let $E_j$ denote the space of all generalised eigenvectors for the eigenvalue $\rho_j$.
    Let $x \in E \setminus \{0\}$.
    Write $x = x_1 + \dots + x_N$ with $x_j \in E_j$.
    Let $1 \leq M \leq N$ be maximal with the property that $x_M \neq 0$.
    Then:
    \begin{enumerate}
        \item There exists a non-negative integer $d$ such that the sequence $\binom{k}{d}^{-1}\rho_M^{-k}\mathcal{L}^{k} x$ converges to an eigenvector for the eigenvalue $\rho_M$ as $k \to \infty$. 
        \item For all $\rho > \rho_M$, the sequence $\rho^{-k} \mathcal{L}^{k}x$ converges to zero as $k \to \infty$.
    \end{enumerate}
\end{lemma}
\begin{proof}
    Consider a Jordan basis of $E$:
    There exists a basis of $E$ consisting of vectors $v_{\ell, p, q} \in E$
    where $\ell \in \{1,\dots,N\}$,
    $p \in \{1,\dots,s_\ell\}$,
    $q \in \{1,\dots,t_{\ell,p}\}$
    such that 
    $\mathcal{L} v_{\ell, p, 1} = \rho_\ell v_{\ell, p, 1}$
    and 
    $\mathcal{L} v_{\ell, p, q + 1} = \rho_\ell v_{\ell, p, q + 1} + v_{\ell, p, q}$
    for all $\ell,p,q$.

    There hence exist unique real numbers $c_{\ell,p,q}$ with
    \[
        x = \sum_{\ell = 1}^M \sum_{p = 1}^{s_\ell} \sum_{q = 1}^{t_{\ell, p}} c_{\ell,p,q} v_{\ell, p,q}.
    \]
    The first sum ranges only from $1$ to $M$, since by assumption on $x$ we have $x_{\ell} = 0$ for all $\ell > M$.
    We hence have for all $k \geq n$:
    \[
        \mathcal{L}^k x = \sum_{\ell = 1}^M \sum_{p = 1}^{s_\ell} \sum_{q = 1}^{t_{\ell, p}} \sum_{d = 0}^{t_{\ell,p}-q} \binom{k}{d} \rho_{\ell}^{k-d} c_{\ell,p,q+d} v_{\ell, p,q}.
    \]
    The second claim is now easy to see.
    For the first claim, let $Q$ be maximal with the property that $c_{M,p,Q} \neq 0$ for some $p$.
    Such a $Q$ exists since $x_M \neq 0$ by assumption.
    Now, observe that the sequence $\binom{k}{Q-1}^{-1} \rho^{-k}_M \mathcal{L}^k x$ converges to a limit as $k \to \infty$, namely 
    \[
        \lim_{k \to \infty} \binom{k}{Q-1}^{-1} \rho^{-k}_M \mathcal{L}^k x 
        =
        \sum_{p = 1}^{s_M} 
            \rho_{M}^{1-Q} c_{M,p, Q} v_{M, p, 1}.
    \]
    This is an eigenvector for the eigenvalue $\rho_M$.
\end{proof}

\section{Recognising Robust Instances}
\label{Section: Computability}

\begin{proposition}\label{Proposition: computability}
    The formulas \eqref{eq: robust escaping condition}, \eqref{eq: robust trappedness condition}, \eqref{eq: robust affine escape}, and \eqref{eq: robust affine trapped} in Theorems \ref{Theorem: characterisation of robust linear instances} and \ref{Theorem: characterisation of robust affine instances} are semidecidable.
\end{proposition}
\begin{proof}
    Polynomial root finding is computable \cite{SpeckerZeros}: given a vector of complex numbers $(a_1,\dots,a_{d}) \in \C^d$ we can compute 
    complex numbers $(z_1,\dots,z_{d}) \in \C^d$ such that 
    $z_1,\dots,z_{d}$ contains all roots of the polynomial
    $X^d + a_{1} X^{d-1} + \dots + a_d$.
    It follows that given a real matrix $A \in \R^{n \times n}$ we can compute the spectrum $\sigma(A)$ as an overt and compact subset of $\C$.
    To semidecide the formula \eqref{eq: robust escaping condition}, observe that the set $[0, +\infty)$ is a computable closed subset of $\C$.
    It follows that $\sigma_{\geq 0}(A) = \sigma(A) \cap [0,+\infty)$ is a computably compact subset of $\C$,
    uniformly in $A$.
    The set $S^{n-1}$ is a computably compact subset of $\R^n$, uniformly in $n$.
    The set 
    \[
        \Set{(\lambda, v) \in \C \times \R^{n}}{Av \neq \lambda v \lor \exists j \in \{1,\dots,m\}:\;B_j v < 0}
    \]
    is easily seen to be a computably open subset of $\C \times \R^n$.
    Semidecidability of \eqref{eq: robust escaping condition} follows.
    Semidecidability of \eqref{eq: robust affine escape} is proved analogously.

    To show that \eqref{eq: robust trappedness condition} is semidecidable,
    we first observe that $\sigma^{\text{odd}}_{> 0}(A)$ is computable as an overt subset of $\R$, uniformly in $A$.
    Indeed, it follows from the fundamental theorem of algebra that a rational interval $(a,b)$ contains a real root of a polynomial $f$ 
    of odd multiplicity if and only if there exist rational numbers $ a < a' < b' < b$ with $f(a') \cdot f(b') < 0$.
    This is clearly semidecidable, yielding the claim of overtness.
    Now, the result follows easily from Proposition \ref{Proposition: overtness and compactness}.
    
    To show that \eqref{eq: robust affine trapped} is semidecidable, we treat both disjunctive clauses separately.
    The second disjunct is semidecided analogously to \eqref{eq: robust trappedness condition}.
    For the first disjunct, we use that we can compute $\sigma(A)$ as a closed subset of $\C$, so that we can semidecide if 
    $1 \notin \sigma(A)$.
    If the we find this to be true, $(A - I)^{-1}$ is well-defined and uniformly computable in $A$ via Gaussian elimination (cf.~\cite{BrattkaZieglerLA}).
    It is then obvious that we can semidecide the condition $B \cdot (A - I)^{-1} b \vlt -\eta$.
\end{proof}
\section{Perturbation Lemmas}
\label{Section: perturbation lemmas}

Our proofs of Theorems \ref{Theorem: characterisation of robust linear instances} and \ref{Theorem: characterisation of robust affine instances} require the existence of certain matrix perturbations with prescribed properties.
We summarise them in this section.
Let $X$ be a metric space and let $x \in X$.
We say that \emph{there exist arbitrarily small perturbations of $x$ with the property $P$}
if for all $\varepsilon > 0$ there exists $\widetilde{x} \in X$ with $d(x, \widetilde{x}) < \varepsilon$
such that $\widetilde{x}$ has the property $P$.

\begin{lemma}\label{Lemma: trapped point perturbation lemma}
    Let $A \in \R^{n \times n}$.
    Let $r \geq 0$.
    Let 
    $\sigma^{\mathrm{odd}}_{> r}({A}) = \left\{\lambda_1,\dots,\lambda_N\right\}$
    where $N \geq 0$.
    For $j = 1, \dots, N$, let $v_j$ be an eigenvector of $A$ for the eigenvalue $\lambda_j$.
    Then there exist arbitrarily small perturbations $\widetilde{A}$ of $A$
    such that
    $\sigma_{>r}\left(\widetilde{A}\right) = \{\lambda_1,\dots,\lambda_N\}$,
    all eigenvalues of $\widetilde{A}$ in $(r, +\infty)$ are simple,
    and $\widetilde{A}$'s eigenspace for $\lambda_j$ is spanned by $v_j$. 
\end{lemma}
\begin{proof}
    Extend the (possibly empty) set $\{v_1,\dots,v_N\}$ to a Jordan basis of 
    $\C^n$.
    Let $Q$ be the corresponding base-change matrix, whose first $N$ columns are $v_1,\dots,v_N$.
    Then we have  
    $
        A = Q J Q^{-1}
    $
    where $J$ is of the form
    $
        J = 
        \begin{pmatrix}
            D & B \\
            0 & C    
        \end{pmatrix}
    $
    with $D \in \R^{N \times N}$ being a diagonal matrix with diagonal entries $\lambda_1,\dots,\lambda_N$,
    $B \in \R^{N \times N}$,
    and a matrix in Jordan normal form $C \in \R^{(n - N) \times (N - n)}$ ,
    all of whose real eigenvalues above $r$ have even multiplicity.

    It remains to show that there exist arbitrarily small perturbations of $C$ such that 
    the perturbed matrices have no real eigenvalues greater than $r$.
    Up to permutation we may assume that $C$ is a block diagonal matrix with diagonal
    entries $C_1, \dots, C_M$, where each $C_i$ is a matrix of the form
    \[
        C_i = 
        \begin{pmatrix}
            \alpha_i  & x_{i,1}   &           &             \\    
                      & \alpha_i  & \ddots    &             \\    
                      &           & \ddots    & x_{i,K_i}   \\  
                      &           &           & \alpha_i        
        \end{pmatrix}
    \]
    with $\alpha_i$ being an eigenvalue of $C$ and $x_{i,j} \in \{0,1\}$ for $j = 1,\dots,K_i$.
    Further $\alpha_i \neq \alpha_j$ for $i \neq j$.

    Now, if $\alpha_i$ is a real eigenvalue of $C$ with $\alpha_i > r$, then the multiplicity of $\alpha_i$ as an eigenvalue of $C$ is even, so that 
    by grouping the diagonal entries into blocks of two,
    the matrix $C_i$ can be written in the form 
    \[
        C_i = 
        \begin{pmatrix}
            C_{i,1}  & Y_{i,1}   &           &                   \\    
                      & C_{i,2}  & \ddots    &                   \\    
                      &          & \ddots    & Y_{i,(K_i - 1)/2}         \\    
                      &          &           & C_{i,(K_i + 1)/2}   \\
        \end{pmatrix}
    \]
    where $Y_{i,j} \in \{0,1\}^{2\times 2}$ and the $C_{i,j}$s are $2\times 2$-matrices of the form 
    \[
        \begin{pmatrix}
            \alpha_i    &   \\
                        &  \alpha_i \\
        \end{pmatrix}
        \;
        \text{or}
        \;
        \begin{pmatrix}
            \alpha_i    &   1\\
                        &  \alpha_i \\
        \end{pmatrix}.
    \] 
    Now, we have $\det \left(C_i - t I_{(K_i+1)\times (K_i+1)}\right) = \prod_{i = 1}^{(K_i+1)/2} \det\left(C_{i,j} - tI_{2\times2}\right)$, 
    so that it suffices to show that the matrices of the above form admit arbitrarily small perturbations without real eigenvalues.
    Now, for any $\varepsilon > 0$ the above matrices can respectively be perturbed into the matrices 
    \[
        \begin{pmatrix}
            \alpha_i     &  -\varepsilon  \\
            \varepsilon  &  \alpha_i      
        \end{pmatrix}
        \;
        \text{and}
        \;
        \begin{pmatrix}
            \alpha_i     &   1\\
            -\varepsilon &  \alpha_i 
        \end{pmatrix}
    \]
    with characteristic polynomials 
    $\chi_1(z) = (\alpha_i - z)^2 + \varepsilon^2$
    and 
    $\chi_2(z) = (\alpha_i - z)^2 + \varepsilon$
    respectively.
    Clearly, neither matrix has any real eigenvalues.
    This proves the claim.
\end{proof}

\begin{lemma}\label{Lemma: robust trapped constraint perturbation lemma}
    Let $v \in \R^n$.
    Let $B_1, B_2 \in \R^n$ be linearly independent.
    Assume that $B_1\cdot v \geq 0$ and $B_2\cdot v \leq 0$.
    Then there exist arbitrarily small perturbations $\widetilde{v}$ of $v$
    with
    $B_{1}\cdot \widetilde{v} > 0$ and 
    $B_{2}\cdot \widetilde{v} < 0$.
\end{lemma}
\begin{proof}
    Let $\varepsilon > 0$.
    Let
        $u_1 = B_{1} - \frac{B_{1} \cdot B_{2}}{B_{2} \cdot B_{2}} B_{2}$
and
        $u_2 = B_{2} - \frac{B_{2} \cdot B_{1}}{B_{1} \cdot B_{1}} B_{1}$.
    Since $B_{1}$ and $B_{2}$ are linearly independent, $u_1$ and $u_2$ are both non-zero.
    We have by construction $B_{1}\cdot u_2 = 0$ and $B_{2}\cdot u_1 = 0$.
    Further, by the Cauchy-Schwarz inequality 
    $B_{1}\cdot\frac{u_1}{B_{1}\cdot B_{1}} > 0$
    and
    $B_{2}\cdot\frac{u_2}{B_{2}\cdot B_{2}} > 0$.
    Let 
    $
        \widetilde{v} = v + \varepsilon/4 \frac{u_1}{B_{1}\cdot B_{1}} - \varepsilon/4 \frac{u_2}{B_{2}\cdot B_{2}}.
    $
    Then $\norm{\widetilde{v} - v} < \varepsilon$,
    $B_{1} \cdot \widetilde{v} > 0$,
    and
    $B_{2} \cdot \widetilde{v} < 0$.
\end{proof}

\begin{lemma}\label{Lemma: perturbing eigenvectors}
    Let $A \in \R^{n \times n}$ be a real matrix.
    Let $v_1,\dots, v_N \in \R^n$ be linearly independent eigenvectors for real eigenvalues $\rho_1,\dots,\rho_N$ of $A$.
    Then for all $\varepsilon > 0$ there exists $\delta > 0$ such that if $w_1,\dots,w_N \in \R^n$
    are vectors with $\norm{v_j - w_j} < \delta$ for all $j$, then there exists a matrix 
    $\widetilde{A}$ with $\norm{A - \widetilde{A}} < \varepsilon$ 
    such that $\widetilde{A}$ has the same eigenvalues as $A$, counted with multiplicity,
    and for all $j \in \{1,\dots,N\}$, we have $\widetilde{A} w_j = \rho_j w_j$.
\end{lemma}
\begin{proof}
    Extend $v_1,\dots, v_N$ to a basis $v_1,\dots,v_n$ of $\R^n$.
    Let $Q$ be the matrix with columns $v_1,\dots,v_n$. 
    Then $Q$ is non-singular and we have 
    $
        A = Q \begin{pmatrix} B   &   C \\ 0   &   D \end{pmatrix} Q^{-1}
    $
    where $B$ is a diagonal matrix with entries $\rho_1,\dots,\rho_N$.
    The set of non-singular matrices is open, and matrix inversion is a continuous operation on this set.
    It follows that for all $\varepsilon > 0$ there exists a $\delta > 0$ such that 
    for all $\widetilde{Q}$, if $\norm{\widetilde{Q} - Q} < \delta$,
    then $\widetilde{Q}$ is nonsingular and 
    $
        \norm{
        \widetilde{Q} \begin{pmatrix} B & C \\ 0 & D\end{pmatrix} \widetilde{Q}^{-1} 
        -
        Q \begin{pmatrix} B   &   C \\ 0   &   D \end{pmatrix} Q^{-1}} < \varepsilon.
    $
    The claim follows.
\end{proof}

\begin{lemma}\label{Lemma: perturbing matrices with 1 in the spectrum}
    Let $A \in \R^{n\times n}$, $B \in \R^{m \times n}$, $b \in \R^n$, $\eta \in \R^m$.
    Assume that $1 \in \sigma(A)$ and that $B \neq 0$.
    Then there exist arbitrarily small perturbations $\widetilde{A}$ of $A$ 
    and $\widetilde{\eta}$ of $\eta$ satisfying the following properties:
    \begin{enumerate}
        \item $1 \notin \sigma\left(\widetilde{A}\right)$.
        \item The unique fixed point of the map $x \mapsto \widetilde{A} x + b$ is not contained in the closed polyhedron 
        $\overline{P}\left(B, \widetilde{\eta}\right) = \Set{x \in \R^n}{Bx \vgeq \eta}$.
        \item A complex number $\lambda \in \C$ with $\left|\lambda\right| > 1$ is an eigenvalue of $A$ with multiplicity $\mu$ if and only if it is an eigenvalue of $\widetilde{A}$ with multiplicity $\mu$.
        Moreover, a vector $v \in \C^n$ is a generalised eigenvector of rank $r$ of $A$ for the eigenvalue $\lambda$ 
        if and only if it is a generalised eigenvector of rank $r$ of $\widetilde{A}$ for the eigenvalue $\lambda$.
    \end{enumerate}
\end{lemma}
\begin{proof}
    Let $\varepsilon > 0$.
    Consider the equation $(A - I)x = -b$.
    The matrix $A - I$ admits an LU-decomposition,
    $
        A - I = PLU 
    $,
    where $P$ is a permutation matrix, $L$ is a lower echelon matrix whose diagonal entries are all equal to $1$,
    and $U$ is an upper echelon matrix.
    We may further assume that exactly the last $r = n - \rank(A - I)$ rows of $U$ are zero.
    By assumption we have $\rank(A - I) < n$, so that $r \geq 1$.
    Let $\widetilde{U}$ be the matrix that results from $U$ by replacing the diagonal entries of all zero rows but the last by $-\varepsilon$.
    For $w \in \R$, let $\widetilde{U}(w)$ be the matrix that results from $\widetilde{U}$ by replacing the last matrix entry
    $\widetilde{U}_{n,n}$ by $w$.
    Whenever $w \neq 0$, the equation $PL\widetilde{U}(w)x = -b$ has a unique solution $x(w)$, which depends (affine) linearly on $w$.
    The set 
    $
        \Set{x(w)}{w \in [-\varepsilon,\varepsilon] \setminus \{0\}}
    $
    is the difference between a straight line and the cylinder $\R^{n-1} \times \left[-|b_n|/\varepsilon, |b_n|/\varepsilon\right]$.
    Since $B$ is not zero, it has a non-zero row, which defines a hyperplane in $\R^n$.
    It is now easy to see that there exists an arbitrarily small perturbation $\widetilde{\eta}$ of $\eta$ 
    such that there exists $w_0 \in [-\varepsilon, \varepsilon]$ with $x(w_0) \notin \overline{P}(B, \widetilde{\eta})$.
    Now, let 
    $\widetilde{A} = PL\widetilde{U}(w_0) + I$.
    It is straightforward to verify that $\widetilde{A}$ has all the claimed properties by construction.
\end{proof}

\section{Robust Escaping Instances of the Linear Problem} 
\label{Section: Robust Escape Linear}

We now prove the first item of Theorem \ref{Theorem: characterisation of robust linear instances}.
    Assume that $\left(A, B\right)$ satisfies \eqref{eq: robust escaping condition}.
    Let $x \in \R^{n}$ be a point with $B x \vgt 0$.
    By standard linear algebra, the matrix $A$ admits a real Jordan decomposition.
    Thus, we can write $\R^n = E_0 \oplus E_{> 0} \oplus E_{\text{rest}}$,
    where $E_{0}, E_{>0}, E_{\text{rest}} \subseteq \R^n$ are linear subspaces of $\R^n$, potentially equal to zero, where $E_0$ is spanned by the generalised eigenvectors for the eigenvalue $0$ of $A$, 
    $E_{> 0}$ is spanned by the generalised eigenvectors for the positive real eigenvalues of $A$, 
    $E_{\text{rest}}$ is invariant under $A$, and the restriction of $A$ to $E_{\text{rest}}$ has no non-negative real eigenvalues.
    
    In particular, we can uniquely write $x = x_{0} + x_{> 0} + x_{\text{rest}}$ 
    for vectors
    $x_{0} \in E_0$, $x_{>0} \in E_{>0}$, $x_{\text{rest}} \in E_{\text{rest}}$.
    By Corollary \ref{Corollary: matrix iterations without positive real eigenvalues}, the sequence 
    $
        \left(B_j A^k (x_{0} + x_{\text{rest}})\right)
    $
    assumes infinitely many non-positive values for all $j$.
    It hence suffices to show that at least one of the sequences 
    $(B_j A^k x_{>0})_k$
    for $j \in \{1,\dots,m\}$
    is eventually non-positive.

    If $x_{> 0} = 0$, then the sequence $(B_j A^k x_{>0})_k$ is zero and hence non-positive for all $j$.
    If $x_{> 0} \neq 0$, then by Lemma \ref{Lemma: linear map with positive real eigenvalues} 
    there exists a positive real eigenvalue $\rho$ of $A$ 
    and a non-negative integer such that the sequence 
    $
        \binom{k}{d}^{-1} \rho^{k} A^k x_{> 0}
    $
    converges to an eigenvector $v$ of $A$ for the eigenvalue $\rho$.

    Since the instance satisfies \eqref{eq: robust escaping condition}, there must exist an index $j \in \{1,\dots,m\}$ with 
    $
        B_j v < 0
    $.
    Now, since multiplication with $B_j$ is continuous and $\binom{k}{d}^{-1} \rho^{k}$ is positive for all $k$, we must have 
    $
        B_j A^k x_{> 0} < 0
    $
    for all sufficiently large $k$.
    This shows that the instance is escaping.
    It must be robustly escaping, since the formula \eqref{eq: robust escaping condition} is semidecidable by Proposition \ref{Proposition: computability}, so that the set of all instances that satisfy the formula is open.

    Now, assume that the formula \eqref{eq: robust escaping condition} is violated for $(A,B)$.
    Then, by definition,
    \[
        \exists \lambda\in\sigma_{\ge0}(A). \exists v\in S^{n-1}.
        A v=\lambda v \land 
        \forall j \in \{1,\dots,m\}. B_j v\geq 0.
    \]
    Thus, there exists a non-negative eigenvalue $\lambda \geq 0$ of $A$ with eigenvector $v$
    such that $B_j v \geq 0$ for all $j$.
    It is easy to see that there exists an arbitrarily small perturbation $\widetilde{A}$ of $A$
    such that $\widetilde{A}$ has a positive real eigenvalue $\widetilde{\lambda} > 0$  
    with $\widetilde{A}v = \widetilde{\lambda} v$.
    Now, we can replace each row $B_j$ of $B$ by $\widetilde{B}_j = B_j + \varepsilon v^T$ for an arbitrarily small $\varepsilon > 0$
    to obtain $\widetilde{B}_j v = \widetilde{B}_j v + \varepsilon \norm{v} > 0$ for all $j$.
    Then the point $x = v$ is trapped in the polyhedron $P\left(\widetilde{B}\right)$ under $\widetilde{A}$, for we have 
    $\widetilde{B} \widetilde{A}^k x = \widetilde{\lambda}^k \widetilde{B}x \vgt 0$ for all $k$.
\section{Robust Trapped Instances of the Linear Problem}
\label{Section: Robust Trapped Linear}

We now prove the second item of Theorem \ref{Theorem: characterisation of robust linear instances}.
    First, assume that \eqref{eq: robust trappedness condition} holds true.
    Then $A$ has a positive real eigenvalue $\lambda > 0$ with odd multiplicity.
    Pick a unit eigenvector $v$ for the eigenvalue $\lambda$.
    Then we either have $Bv \vgt 0$ or $B(-v) \vgt 0$.
    In the first case, the point $v$ is trapped in $P(B)$,
    in the second case, the point $-v$ is trapped.
    This shows that the instance is trapped.
    It must be robustly trapped, since the formula \eqref{eq: robust trappedness condition} is semidecidable by Proposition \ref{Proposition: computability}, so that the set of all instances that satisfy the formula is open.

    Now, assume that \eqref{eq: robust trappedness condition} is false.
    We will construct arbitrarily small perturbations of the instance that are escaping.
    By assumption we have:
    \begin{equation}\label{eq: negation of robust trappedness}
        \forall \lambda\in\sigma^{\text{odd}}_{>0}(A).\exists v\in S^{n-1}.
        A v=\lambda v \land 
        \exists j. B_j v \leq 0 \land \exists j. B_j v \geq 0.
    \end{equation}
    
    Let $0< \rho_1 < \dots < \rho_N$, for $N \geq 0$, denote the positive real eigenvalues of $A$ with odd algebraic multiplicity.
    For each $\ell \in \{1,\dots,N\}$, pick a witness $v_\ell$ of \eqref{eq: negation of robust trappedness} for $\rho_\ell$,
    \textit{i.e.}, an eigenvector for the eigenvalue $\rho_\ell$ such that $B_{j_1} v_\ell \leq 0$ and 
    $B_{j_2} v_\ell \geq 0$ for some indexes $j_1, j_2 \in \{1,\dots,m\}$.

    By Lemma \ref{Lemma: trapped point perturbation lemma}, there exists an arbitrarily small perturbation
    $\widetilde{A}_1$ of $A$ such that the positive real eigenvalues of $\widetilde{A}_1$ are 
    $\rho_1,\dots,\rho_N$,
    each $\rho_\ell$ is a simple eigenvalue of $\widetilde{A}_1$, and the corresponding eigenspace is spanned by 
    $v_\ell$.

    Assume that at least two of the rows $B_1,\dots, B_m$ of $B$ are linearly independent.
    We will return to the (easy) case where all of these vectors are pairwise linearly dependent at the end of the proof.
    We will show that there exist arbitrarily small perturbations $w_1, \dots, w_m$
    of the vectors $v_1,\dots, v_m$ such that for all $\ell$ there exist $j_1$ and $j_2$ such that 
    $B_{j_1} w_\ell < 0$ and $B_{j_2} w_\ell > 0$.
    Let $j \in \{1,\dots,m\}$.
    Then by construction there exist $j_1$ and $j_2$ such that 
    $B_{j_1} v_\ell \leq 0$ and $B_{j_2} v_\ell \geq 0$.
    If $B_{j_1}$ and $B_{j_2}$ are linearly independent, then we can apply Lemma \ref{Lemma: robust trapped constraint perturbation lemma} to find an arbitrarily small perturbation $w_\ell$ of $v_\ell$ with the desired property.
    If they are linearly dependent, then $B_{j_1} = c B_{j_2}$ for some constant $c$.
    If $c < 0$ then for all $x \in \R^n$, the inequality $B_{j_1} x > 0$ implies $B_{j_2} x < 0$, so that the polyhedron is empty and the instance is escaping for trivial reasons.
    Hence, let us assume that $c \geq 0$. 
    Then we must have $B_{j_1} v_\ell = B_{j_2} v_\ell = 0$.
    Now, since not all of the vectors $B_1, \dots, B_m$ are pairwise linearly dependent, there must exist $j_3$ such that $B_{j_3}$ is linearly independent of $B_{j_1}$.
    Since $B_{j_1} v_{\ell} \geq 0$ and $B_{j_1} v_{\ell} \leq 0$, we can apply Lemma \ref{Lemma: robust trapped constraint perturbation lemma} to $B_{j_1}$ and $B_{j_3}$ to obtain the desired perturbation.
    This proves the claim.

    By Lemma \ref{Lemma: perturbing eigenvectors} we can further perturb the matrix $\widetilde{A}_1$ into a matrix $\widetilde{A}_2$ such that the positive real eigenvalues of $\widetilde{A}_2$ are $\rho_1,\dots,\rho_N$, all of these eigenvalues are simple, and their eigenspaces are spanned by the vectors $w_1,\dots, w_N$ respectively.
    We can write $\R^n$ as a direct sum 
    $
        \R^n = \R w_1 \oplus \dots \oplus \R w_N + R
    $
    where $R$ is a linear subspace of $\R^n$, invariant under $\widetilde{A}_2$, and the restriction of $\widetilde{A}_2$ to $R$ has no positive real eigenvalues. 
    Let $x \in \R^n$.
    Then we can write 
    $
        x = \alpha_1 w_1 + \dots + \alpha_N w_N + r
    $
    for unique real numbers $\alpha_1,\dots, \alpha_N$ and a vector $r \in R$.

    We have by definition for all $j$ and $k$:
    \[
        B_j \widetilde{A}_2^k x = \left(\alpha_1 B_j w_1\right) \rho_1^k + \dots + \left(\alpha_N B_j w_N\right) \rho_N^k + B_j \widetilde{A}_2^k r.
    \]
    By Corollary \ref{Corollary: matrix iterations without positive real eigenvalues}, the sequence
    $B_j \widetilde{A}_2^k r$ assumes non-positive values infinitely often for all $j$.

    It hence suffices to show that the sequence 
    $
        \left(\alpha_1 B_j w_1\right) \rho_1^k + \dots + \left(\alpha_N B_j w_N\right) \rho_N^k
    $
    assumes non-positive values infinitely often.
    If $\alpha_1 = \dots = \alpha_N = 0$, then this is obvious.
    If not all $\alpha_{\ell}$s are zero, let $L$ be the largest integer such that $\alpha_L \neq 0$.

    Now, by construction of $w_L$ there exists $j$ such that $\alpha_{L} B_j w_{L} < 0$.
    Since $\rho_{L} > \rho_j$ for all $j < L$ and $\alpha_j = 0$ for $j > L$, the sequence 
    $
        \left(\alpha_1 B_j w_1\right) \rho_1^k + \dots + \left(\alpha_N B_j w_N\right) \rho_{N}^k
    $
    is eventually negative.
    Since $B_j \widetilde{A}_2^k r$ is non-positive infinitely often, 
    the sequence $(B_j \widetilde{A}_2^k x)_k$ assumes a negative value infinitely often, so that 
    $x$ escapes the polyhedron $P(B)$ under iteration of $\widetilde{A}_2$.

    It remains to consider the possibility that $B_1,\dots,B_m$ are all linearly dependent.
    It follows as before that either the polyhedron is empty or we must have $B_j v_\ell = 0$ for all $j$ and all $\ell$.
    In the first case, the instance is escaping, so it remains to consider the second case.
    Any $x \in \R^n$ can be decomposed as above:
    $
        x = \alpha_1 v_1 + \dots + \alpha_N v_N + r.
    $
    We have by definition for all $j$, using that $B_j v_\ell$ vanishes for all $\ell$:
    $
        B_j \widetilde{A}_2^k x = B_j \widetilde{A}_2^k r.
    $
    It follows as above that $B_j \widetilde{A}_2^k x$ assumes non-positive values infinitely often.

\section{Robust Escaping Instances of the Affine Problem}
\label{Section: Robust Escape Affine}

We now prove the first item of Theorem \ref{Theorem: characterisation of robust affine instances}.
    Assume that the instance satisfies the formula.
    We show that every $x \in P(\widehat{B})$ escapes $P(\widehat{B})$ under iteration of $\widehat{A}$.
    We have
    $\R^{n + 1} = E_{\geq 1} \oplus E_{< 1} \oplus E_{\text{rest}}$
    where $E_{\geq 1}$ is the space spanned by the generalised eigenvectors of $\widehat{A}$ for eigenvalues in $(1, + \infty)$,
    $E_{< 1}$ is the space spanned by the generalised eigenvectors of $\widehat{A}$ for eigenvalues in $(0, 1)$,
    and $E_{\text{rest}}$ is invariant under $A$, and the restriction of $A$ to $E_{\text{rest}}$ has no positive real eigenvalues.

    Let $x \in P(\widehat{B})$.
    Then we can write $x = x_{\geq 1} + x_{< 1} + x_{\text{rest}}$
    for unique $x_{\geq 1} \in E_{\geq 1}$, $x_{< 1} \in E_{< 1}$, and $x_{\text{rest}} \in E_{\text{rest}}$.
    Since the $(n + 1)^{\text{st}}$ component of $x$ is strictly greater than $0$, we must have $x_{\geq 1} \neq 0$.
    Let $1 \leq \rho_1 < \dots < \rho_N$ denote the real eigenvalues of $A$ above $1$.
    We can further write 
    $E_{\geq 1} = E_1 \oplus \dots \oplus E_N$
    where $E_\ell$ is the space of generalised eigenvectors for the eigenvalue $\rho_\ell$.
    Then we have $x_{\geq 1} = x_1 + \dots + x_N$ for unique vectors $x_\ell \in E_{\ell}$.
    Let $M$ be the largest integer such that $x_M \neq 0$.
    Then by Lemma \ref{Lemma: linear map with positive real eigenvalues}, 
    there exists $d \geq 0$ such that 
    the sequence 
    $\binom{k}{d}^{-1}\rho_M^{-k} \widehat{A}^k\left(x_{\geq 1} + x_{< 1}\right)$
    converges to an eigenvector $v$ for the eigenvalue $\rho_M$.
    By assumption, there exists $j$ such that $\widehat{B}_j v < 0$.
    It follows that the sequence $\widehat{B}_j\widehat{A}^k\left(x_{\geq 1} + x_{< 1}\right)$ is eventually negative.
    Now, by Corollary \ref{Corollary: matrix iterations without positive real eigenvalues}, the sequence 
    $\widehat{B}_j \widehat{A}^k x_{\text{rest}}$
    assumes a non-positive value infinitely often.
    Overall, it follows that the sequence $\widehat{B}_j \widehat{A}^k x$ assumes infinitely many negative values.
    In particular, the point escapes the polyhedron under iteration of $\widehat{A}$.
    It follows that the point robustly escapes the polyhedron, since the formula is semidecidable by Proposition \ref{Proposition: computability}, so that the set of all instances that satisfy the formula is open.

    Now assume that the instance does not satisfy the formula.
    Then we have:
    \[
        \exists \lambda \in \sigma_{\geq 1}\left(\widehat{A}\right).
        \exists \widehat{v} \in S^{n}. 
        \left(
            \widehat{A}\widehat{v} = \lambda \widehat{v}
            \;\land\;
            \forall j. \widehat{B}_j \widehat{v} \geq 0.
        \right)
    \]
    We show that there exists a trapped point under arbitrarily small perturbations of $A$ and $B$.
    Let $\lambda$ and $\widehat{v}$ witness the above.
    Let $z$ be the $(n+1)^{\text{st}}$ entry of $\widehat{v}$.
    We have 
    $z = \widehat{B}_{m+1} \widehat{v} \geq 0$.
    If $z > 0$, then we must have $\lambda = 1$.
    It follows that the point $x^* = \left(\widehat{v}_1/z,\dots,\widehat{v}_n/z\right)$ is a fixed point of the map $x \mapsto Ax + b$.
    Further, we have for all $j \in \{1,\dots,m\}$:
    $
        \widehat{B}_j \widehat{v} = z \left(B_j x^* - \eta_j\right) \geq 0.
    $
    Since $z > 0$ we obtain 
    $B_j x^* - \eta_j \geq 0$.
    Up to arbitrarily small perturbation of $B$, these inequalities are strict.
    Then the map $x \mapsto Ax + b$ has a fixed point in the polyhedron, which in particular is trapped.

    If $z = 0$, then $v = (\widehat{v}_1,\dots,\widehat{v}_n)$ is an eigenvector of $A$ for the eigenvalue $\lambda$.
    It is easy to see that there exists an arbitrarily small perturbation $\widetilde{A}$ of $A$
    such that $1 \notin \sigma\left(\widetilde{A}\right)$ and $v$ is an eigenvector of $\widetilde{A}$ for an eigenvalue 
    $\widetilde{\lambda} > 1$.
    Then the map $f(x) = \widetilde{A}x + b$ has a unique fixed point $x^*$.
    Consider the point $x = x^* + v$.
    Then we have 
    $
        f^{k}(x) = x^* + \widetilde{\lambda}^k v.
    $
    Now, letting $\widetilde{B}$ be the matrix with columns $\widetilde{B}_j = B_j + \delta v^T$ for arbitrarily small $\delta > 0$,
    we have $\widetilde{B} v \vgt 0$, 
    so that $\widetilde{B} f^{k}(x) \vgt \eta$ for all sufficiently large $k$.
\section{Robust Trapped Instances of the Affine Problem}
\label{Section: Robust Trapped Affine}

We now prove the second item of Theorem \ref{Theorem: characterisation of robust affine instances}.
    Assume that the instance satisfies the formula.
    If the instance satisfies the clause
    $
        1 \notin \sigma(A) \land B \cdot (A - I)^{-1} b \vlt -\eta 
    $
    then the map $x \mapsto Ax + b$ has a unique fixed point, given by 
    $x^* = -(A - I)^{-1} b$.
    We have 
    $
        B x^* = -B \cdot (A - I)^{-1} b \vgt \eta.
    $
    Thus, $x^*$ is trapped in the polyhedron under $x\mapsto Ax + b$.

    If the instance does not satisfy the first clause, then it must satisfy
    \[
        \exists \lambda\in\sigma^{\text{odd}}_{> 1}(A).\forall v\in S^{n-1}.
        \left(
        A v=\lambda v\;\Longrightarrow\;
        \left(\forall j.B_j v>0 \lor \forall j.B_j v<0\right)
        \right).
    \]
    Let $\lambda > 1$ witness the above.
    Then there exists an eigenvector $v$ for $\lambda$ satisfying $B v \vgt 0$.
    Let $E_1$ be the space of generalised eigenvectors of $A$ for the eigenvalue $1$.
    Let $E_2$ be the space spanned by the generalised eigenvectors of $A$ for all other eigenvalues.
    Then $E_1$ and $E_2$ are invariant under $A$ and we have 
    $
        \R^n = E_1 \oplus E_2
    $.
    In particular, there exist unique linear projections $\pi_j \colon \R^n \to E_j$ such that we have 
    $x = \pi_1(x) + \pi_2(x)$
    for all $x \in \R^n$.
    Let $f(x) = Ax + b$.
    Since the restriction of $A$ to the space $E_2$ does not have the eigenvalue $1$, there exists 
    a unique $x^* \in E_2$ with $f(x^*) = x^* + x_1$, where $x_1 \in E_1$.
    Consider the point $x = x^* + v$.
    Then we have
    $
        f^{k}(x) = A^k v + f^{k}(x^*) = A^k v + x^* + A^k x_1 = \lambda^k v + x^* + A^k x_1.
    $
    Since $x_1 \in E_1$, the norm of the vector $A^k x_1$ is a polynomial in $k$.
    It follows that 
    $
        \lambda^{-k} f^{k}(x) \to v
    $
    as $k \to \infty$.
    Hence,
    $\lambda^{-k} B f^{k}(x) \to B v \vgt 0$
    as $k \to \infty$.
    It follows that $B f^{k}(x) \vgt \eta$ for all sufficiently large $k$.
    Hence, the instance is trapped.
    It must be robustly trapped, since the formula is semidecidable by Proposition \ref{Proposition: computability}, so that the set of all instances that satisfy the formula is open.
    
    Now, assume that the formula does not hold true.
    Then we have:
    \begin{align*}
        &\left( 1 \in \sigma(A) \lor B \cdot (A - I)^{-1} b \not\vlt -\eta \right)\\
        &\land
        \forall \lambda\in\sigma^{\text{odd}}_{> 1}(A).\exists v\in S^{n-1}.
        \left(
        A v=\lambda v \land 
        \exists j. B_j v\leq0 \land \exists j. B_j v\geq0
        \right).
    \end{align*}
    If $B = 0$, then it is easy to see that $P(B,\eta)$ is empty up to an arbitrarily small perturbation of $B$ and $\eta$, so that the instance escapes for trivial reasons.
    We may hence assume that $B \neq 0$.
    If $1 \notin \sigma(A)$, then by the above formula, the unique fixed point $x^*$ of the map $x \mapsto Ax + b$ is not contained in $P(B, \eta)$.
    Thus, there exists an arbitrarily small perturbation $\widetilde{\eta}$ of $\eta$ such that $x^*$ is not contained in 
    $\overline{P}\left(B, \widetilde{\eta}\right)$.
    If $1 \in \sigma(A)$, then by Lemma \ref{Lemma: perturbing matrices with 1 in the spectrum}, there exist arbitrarily small perturbations $\widetilde{A}_1$ of $A$
    and $\widetilde{\eta}$ of $\eta$ such that 
    $1 \notin \sigma\left(\widetilde{A}_1\right)$,
    the unique fixed point of the map $x \mapsto \widetilde{A}_1 x + b$ is not contained in the closed polyhedron $\overline{P}\left(B, \widetilde{\eta}\right)$,
    a complex number $\lambda \in \C$ with $\left|\lambda\right| > 1$ is an eigenvalue of $A$ with multiplicity $\mu$ if and only if it is an eigenvalue of $\widetilde{A}_1$ with multiplicity $\mu$,
    and a vector $v \in \C^n$ is a generalised eigenvector of rank $r$ of $A$ for the eigenvalue $\lambda$ 
        if and only if it is a generalised eigenvector of rank $r$ of $\widetilde{A}_1$ for the eigenvalue $\lambda$.
    In either case, there exist arbitrarily small perturbations $\widetilde{A}_1$ of $A$ and $\widetilde{\eta}$ of $\eta$ with the above properties.

    Let $\sigma^{\text{odd}}_{> 1}(\widetilde{A}_1) = \{\rho_1,\dots,\rho_N\}$ where $\rho_1 < \dots < \rho_N$.
    Then for all $\ell \in \{1,\dots,N\}$ there exists an eigenvector $v_{\ell}$ of $\widetilde{A}_1$ for $\rho_{\ell}$ and a $j \in \{1,\dots, m\}$ such that $B_j v_{\ell} \leq 0$.
    By Lemma \ref{Lemma: trapped point perturbation lemma} there exists an arbitrarily small perturbation $\widetilde{A}_2$ of $\widetilde{A}_1$ such that $\sigma\left(\widetilde{A}_2\right)_{> 1} = \{\rho_1,\dots,\rho_N\}$,
    and the eigenspace for $\rho_{\ell}$ is spanned by $v_{\ell}$ for all $\ell$.
    Since the unique fixed point of $x \mapsto \widetilde{A}_1x + b$ depends continuously on $\widetilde{A}_1$ and the polyhedron $\overline{P}\left(B, \widetilde{\eta}\right)$ is closed, we can further ensure that the unique fixed point of the map $x \mapsto \widetilde{A}_2x + b$ is not contained in the closed polyhedron $\overline{P}\left(B, \widetilde{\eta}\right)$.

    Assume that $B$ has at least two linearly independent rows. 
    Then, arguing as in the proof of Theorem \ref{Theorem: characterisation of robust linear instances}(\ref{Item: robust linear trapped}),
    we can (without loss of generality) use Lemmas \ref{Lemma: robust trapped constraint perturbation lemma} and 
    \ref{Lemma: perturbing eigenvectors} to find an arbitrarily small perturbation $\widetilde{A}_3$ of $\widetilde{A}_2$ such that 
    $\widetilde{A}_3$ has the same eigenvalues as $\widetilde{A}_2$, counted with multiplicities
    and 
    for all $\ell \in \{1,\dots,N\}$ the eigenspace of $\widetilde{A}_3$ for $\rho_{\ell}$ is spanned by a vector $w_{\ell}$ with 
    $B_j w_{\ell} < 0$ for some $j$.
    As above, we can further ensure that the unique fixed point of the map $x \mapsto \widetilde{A}_3x + b$ is not contained in the closed polyhedron $\overline{P}\left(B, \widetilde{\eta}\right)$.

    Now, let $x \in \R^n$. Let $x^*$ be the unique fixed point of the map that sends $x$ to $\widetilde{A}_3x + b$.
    Let $E_0$ denote the space spanned by the generalised eigenvectors of $\widetilde{A}_3$ for eigenvalues in $\C \setminus (0,+\infty)$.
    Let $E_1$ denote the space spanned by the generalised eigenvectors of $\widetilde{A}_3$ for eigenvalues in $(0,1)$.
    Since $1$ is not an eigenvalue of $\widetilde{A}_3$, the spectrum of $\widetilde{A}_3$ is covered by the sets $(0,1)$,
    $(1,+\infty)$, and $\C \setminus (0,+\infty)$.
    It follows that there exist unique real numbers $c_1,\dots,c_N$ and vectors $r_{0} \in E_0$, $r_{1} \in E_1$ such that 
    \[
        x = x^* + c_1 w_{1} + \dots + c_N w_{N} + r_{0} + r_{1}.
    \]
    Let $f(x) = \widetilde{A}_3x + b$.
    Then we have:
    \[
        f^{k}(x) = x^* + \widetilde{A}_3^k r_1 + c_1 \rho_1^k w_1 + \dots + c_N \rho_N^k w_N + \widetilde{A}^k_3 r_0.
    \]
    Hence:
    \[
        B f^{k}(x) = B\left(x^* + \widetilde{A}_3^k r_1\right) + B\left(c_1 \rho_1^k w_1 + \dots + c_N \rho_N^k w_N\right) + B \widetilde{A}_3^k r_0.
    \]
    Now, the sequence $x^* + \widetilde{A}_3^k r_1$ converges to $x^*$ with $Bx^* \vlt \widetilde{\eta}$.
    By Corollary \ref{Corollary: matrix iterations without positive real eigenvalues}, the sequence $B_{j} \widetilde{A}_3^k r_0$ assumes a non-positive value infinitely often for all $j$.
    It is easy to see that the sequence $B_{j}\left(c_1 \rho_1^k w_1 + \dots + c_N \rho_N^k w_N\right)$ is either zero or eventually negative for all $j$.
    It follows that $B_{j} f^{k}(x) \vlt \widetilde{\eta}$ for all sufficiently large $k$.

    The case where all rows of $B$ are linearly dependent is treated as in the proof of Theorem \ref{Theorem: characterisation of robust linear instances}(\ref{Item: robust linear trapped}).

\section{On the Measure of the Boundary Instances}
\label{Section: Measure Boundary}

We now prove Theorem \ref{Theorem: measure of boundary instances} for the Linear Universal Escape Problem.
    The proof for the Affine Universal Escape Problem is analogous.
    We may assume without loss of generality that $A$ is non-singular, that $B$ has full rank, and that all eigenvalues of $A$ are simple,
    for the set of remaining problem instances already has measure zero.
    By Theorem \ref{Theorem: characterisation of robust linear instances}(\ref{Item: robust linear escaping})
    an escaping instance is a boundary instance if and only if it satisfies the formula 
    $
        \exists\lambda\in\sigma_{\ge0}(A).\exists v\in S^{n-1}.
        A v=\lambda v \land B v \vgeq 0.
    $
    The set of all boundary escaping instances is hence contained in the projection of the set 
    \[
        \Set{\left(A, B, \lambda, v\right) \in \R^{n \times n} \times \R^{m \times n}\times \R \times \R^n}
            {\norm{v}_2^2 = 1, \lambda \geq 0, A v = \lambda v, B v \vgeq 0}.
    \]
    If a problem instance further satisfies $\lambda > 0$ and $B v \vgt 0$, then it is trapped.
    Hence, the set of all boundary escaping instances is contained in the projection of the set 
    \begin{align*}
        &\Set{\left(A, B, \lambda, v\right) \in \R^{n \times n} \times \R^{m \times n}\times \R \times \R^n}
            {\norm{v}_2^2 = 1, \lambda = 0, A v = \lambda v}\\
        &\cup 
        \bigcup_{j = 1}^m 
            \Set{\left(A, B, \lambda, v\right) \in \R^{n \times n} \times \R^{m \times n}\times \R \times \R^n}
            {\norm{v}_2^2 = 1, A v = \lambda v, B_j v = 0}.
    \end{align*}
    Using that $A$ and $B$ have full rank, we find that the above is a real algebraic set of dimension $n^2 + mn - 1$.
    Its projection to $\R^{n\times n} \times \R^{m \times n}$ is hence a semi-algebraic set of codimension at most one \cite[Lemma 5.30]{BasuPollackRoy}.
    It follows that the set has zero measure.

    Now, assume that $(A,B)$ is a boundary trapped instance. Then by Theorem \ref{Theorem: characterisation of robust linear instances}(\ref{Item: robust linear trapped}), the instance satisfies the formula
    \[
        \forall \lambda\in\sigma^{\text{odd}}_{>0}({A}).\exists v\in S^{n-1}.
            A v=\lambda v \land \exists j. B_j v\leq0 \land \exists j. B_j v\geq0.
    \]
    Since we assume that $A$ has only simple eigenvalues, $\sigma^{\text{odd}}_{>0}({A}) = \sigma_{>0}(A)$.
    Since $A$ and $B$ have full rank, $A$ must have a positive real eigenvalue, for otherwise the instance is escaping by Corollary \ref{Corollary: matrix iterations without positive real eigenvalues}.
    The instance hence satisfies the formula
    \[
        \exists \lambda\in\sigma_{>0}({A}). \exists v\in S^{n-1}.
        A v=\lambda v \land \exists j.  B_j v\leq0 \land \exists j. B_j v\geq0.
    \]
    It is hence contained in the projection of the set 
    \[
        \Set{(A,B,\lambda, v) \in \R^{n \times n} \times \R^{m \times n}\times \R \times \R^n}{\norm{v}_2^2 = 1, A v=\lambda v, \exists j.{B}_j v\leq0}.
    \]
    Now, if for all $\lambda \in \sigma_{>0}(A)$ and all $v$ with $Av = \lambda v$ there exists $j$ such that $B_j v < 0$, then the instance is robustly escaping by the proof of Theorem \ref{Theorem: characterisation of robust linear instances}(\ref{Item: robust linear escaping}).
    It follows that the instance is contained in the projection of the set
    \[
        \Set{(A,B,\lambda, v) \in \R^{n \times n} \times \R^{m \times n}\times \R \times \R^n}{\norm{v}_2^2 = 1, A v=\lambda v, \exists j.{B}_j v=0}.
    \] 
    This is again a real algebraic set of dimension $n^2 + mn - 1$.
    Its projection is hence a semi-algebraic set of dimension at most $n^2 + mn - 1$.
    In particular, it has zero measure.
\bibliographystyle{alpha}
\bibliography{loops-cie-2026}
\nocite{*}

\end{document}